\documentclass[preprint,12pt]{elsarticle}

\usepackage{graphicx}
\usepackage{threeparttable}
\usepackage{multirow}
\usepackage{booktabs}
\usepackage{longtable}
\usepackage{pdflscape}
\usepackage{caption}
\usepackage{amsmath,amssymb,amsthm}
\usepackage{lineno}

\newtheorem{proposition}{Proposition}
\newtheorem{assumption}{Assumption}

\journal{Computational Statistics \& Data Analysis}

\begin{document}

\begin{frontmatter}

\title{Recovering the Target Hazard Ratio Under Nonproportional Hazards Induced by an Omitted Covariate: Simulation-based Approach}

\author[nci]{Jong-Hyeon Jeong\corref{cor1}}
\ead{jong-hyeon.jeong@nih.gov}

\cortext[cor1]{Corresponding author.}

\address[nci]{Biometric Research Program, DCTD, National Cancer Institute, Rockville, MD, USA}

\begin{abstract}
When an omitted covariate whose inclusion would balance proportional hazards is excluded from a proportional hazards model, bias in the estimated treatment effect may arise. The omitted covariate may represent an unobservable biomarker status, an overlooked stratification factor, or a strong continuous prognostic factor. In this paper, we propose a simple simulation-based approach for recovering the target hazard ratio for treatment effect, defined as the hazard ratio from the correctly specified proportional hazards model that includes the omitted covariate. Our approach identifies the target hazard ratio value that generates a band of survival curves enclosing the observed survival curve estimates most, under the minimal assumptions of a Weibull baseline event-time distribution, (a mixture of) uniform censoring, and unit variance of the omitted covariate. Simulation studies indicate that the proposed method recovers the target hazard ratio reasonably well under practical scenarios involving a broad range of true hazard ratios, regardless of the true distribution of the omitted covariate. Consistency of the empirical estimates from the proposed procedure to the true values is proved. We illustrate the proposed method using data from a phase III breast cancer clinical trial.
\end{abstract}


\end{frontmatter}


\section{Introduction}

The proportional hazards assumption is crucial when the Cox model \cite{cox1972,cox1976} is used in data analysis. It is well known that, under model misspecification, the score test statistic based on the partial likelihood may lose efficiency \cite{struthers1986,oakesjeong1998}, and the partial likelihood estimator may be substantially biased \cite{struthers1986,henderson1999,harrington1982,gail1984,huber1988,byar1988,ford1995,aalen2015,lancaster1990,schmoor1997}. In practice, however, the proportional hazards assumption may be violated for a variety of reasons, including an unpredictable onset of treatment effect, patient noncompliance with the protocol therapy, or omission of a covariate whose inclusion would restore proportional hazards.

As a real example of nonproportional hazards, consider a clinical trial conducted by the National Surgical Adjuvant Breast and Bowel Project (NSABP), a National Cancer Institute (NCI) cooperative group. In this study, patients with primary breast cancer, negative axillary nodes, and estrogen receptor--positive tumors were randomized to receive either tamoxifen (hormonal therapy) or placebo following surgery \cite{fisher1989}. In this trial, the proportional hazards assumption for disease-free survival (DFS) between the treatment groups was violated, with a significant $p$-value of 0.00018. Figure \ref{fig:b14_plots} presents the Kaplan--Meier estimates of the DFS survival distribution (Figure \ref{fig:b14_plots}(a)), the censoring distribution (Figure \ref{fig:b14_plots}(b)), and the smoothed scaled Schoenfeld residuals \cite{schoenfeld1982} (Figure \ref{fig:b14_plots}(c)). The residual plot suggests that the estimated treatment effect, quantified by the hazard ratio, diminishes over time.

\begin{figure}[htbp]
  \makebox[\textwidth][c]{\includegraphics[width=1.6\textwidth]{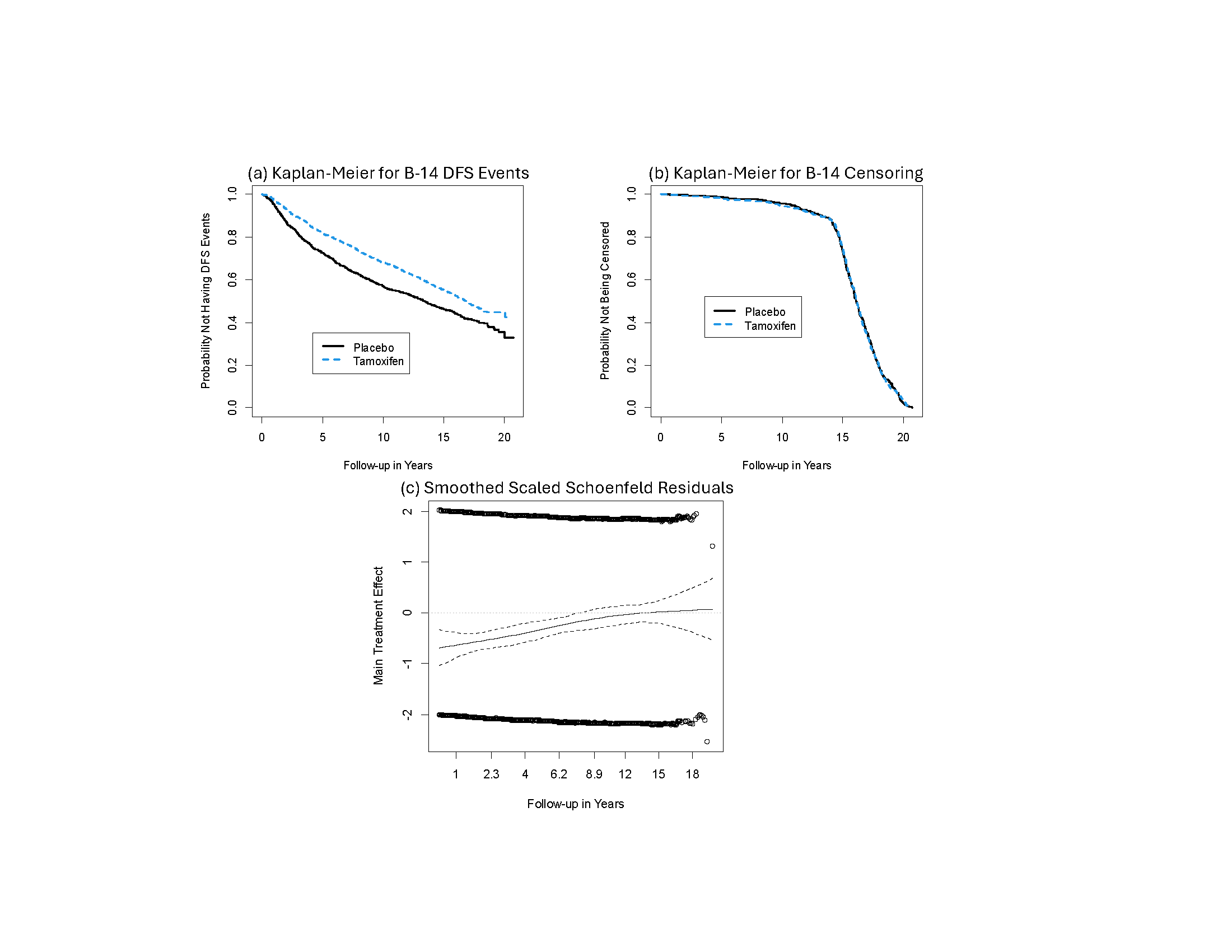}}%
    \vspace{-0.7in}
    \caption{Kaplan--Meier estimates of the disease-free survival (DFS) distribution and censoring distribution, and the smoothed scaled Schoenfeld residual plot, for the B-14 data set.}
    \label{fig:b14_plots}
\end{figure}

In this paper, we propose a simple and practical approach for recovering the target hazard ratio, defined as the hazard ratio from the correctly specified proportional hazards model that includes the omitted covariate. Under the assumptions of a Weibull baseline event-time distribution, unit variance of the omitted covariate, and a univariate frailty term for the unobservable omitted covariate, our method identifies the value that would have generated a band of survival curves containing the observed survival curves as most as possible, as will be formally defined in Section 6. 

The remainder of the paper is organized as follows. In Section 2, we define the correctly specified and misspecified proportional hazards models. In Section 3, we review and extend bias expressions under commonly used univariate frailty models. In Section 4, we review existing approaches to bias correction under model misspecification. In Section 5, we present simulation studies evaluating the extent to which frailty-based modeling and the existing analytic method recover the true hazard ratio under practical scenarios. A simulation-based approach based on so-called Tightest Survival Band (TSB) to recover the target hazard ratio is presented in Section 6, together with numerical examples and extensive simulation studies. In Section 7, we illustrate the proposed method using the NSABP B-14 breast cancer trial data. Section 8 concludes with a discussion, and Appendix A provides a formal consistency result for the proposed TSB estimator. Supplementary tables and figures referenced throughout are provided in the Supplementary Materials.

\section{Models}

Suppose the event time follows the Cox proportional hazards model \cite{cox1972,cox1976},
\begin{equation}
    h(t;X,Z)=h_0(t)\exp(\beta X+\theta Z),
\end{equation}
or, equivalently, in terms of the survival function,
\begin{equation}
    S^{(I)}(t;X,Z)=S_0(t)^{\exp(\beta X+\theta Z)},
    \label{eqn;1}
\end{equation}
where $h_0(t)$ is the baseline hazard function, $X$ is a binary treatment indicator, and $Z$ is a continuous covariate. When both $X$ and $Z$ are observed, model \eqref{eqn;1} is correctly specified and the proportional hazards assumption holds; we refer to this as Model I.

Throughout the paper, we also consider two related models. In Model II, the omitted covariate $Z$ is unobserved and is instead absorbed into the baseline distribution through the univariate frailty term $W=\exp(\theta Z)$, thereby forcing a proportional hazards structure \cite{oakesjeong1998,henderson1999}. Model III is the marginal model, or univariate frailty model \cite{vaupel1979}, which explicitly accounts for the distribution of the omitted covariate. In terms of the survival function, Models II and III are given by
\begin{equation}
    S^{(II)}(t;X)=E_W[\exp\{-H_0(t)W\}]^{e^{\beta X}}
    =L(H_0(t))^{e^{\beta X}},
    \label{eqn;1.mis}
\end{equation}
and
\begin{equation}
    S^{(III)}(t;X)=E_W[\exp\{-e^{\beta X}H_0(t)W\}]
    =L(e^{\beta X}H_0(t)),
    \label{eqn;2}
\end{equation}
respectively, where $H_0(t)=-\log\{S_0(t)\}$ is the baseline cumulative hazard function and
\[
L(v)=\int \exp(-vw)\,dF_W(w)
\]
is the Laplace transform of the frailty distribution $F_W(\cdot)$. Henderson and Oman \cite{henderson1999} showed that, for small $\beta$, the survival function under the misspecified model that omits $Z$ is well approximated by Model II.

\section{Bias Factors Under Univariate Frailty Models}

Under Model III in \eqref{eqn;2}, using the fact that the marginal density is the negative derivative of the marginal survival function, we have
\begin{eqnarray*}
f(t;X)
&=&
-\frac{dS(t;X)}{dt} \\
&=&
E_W[W\exp\{-H_0(t)We^{\beta X}\}]\,e^{\beta X}\frac{dH_0(t)}{dt} \\
&=&
-L^{\prime}(H_0(t)e^{\beta X})\,e^{\beta X}\frac{dH_0(t)}{dt}.
\end{eqnarray*}
Therefore, the marginal hazard function can be written as
\[
h^{(m)}(t;X)=
-\frac{L^{\prime}(H_0(t)e^{\beta X})}{L(H_0(t)e^{\beta X})}
\,e^{\beta X}\frac{dH_0(t)}{dt}.
\]
Thus, the marginal log-hazard ratio under Model III can be expressed as \cite{balan2020,therneau2000}
\begin{equation}
\log HR^{(m)}(t)
=
\log\frac{h^{(m)}(t;X=1)}{h^{(m)}(t;X=0)}
=
\beta
+
\log
\frac{L^{\prime}(H_0(t)e^{\beta})L(H_0(t))}
     {L(H_0(t)e^{\beta})L^{\prime}(H_0(t))}.
\label{eqn;6.1}
\end{equation}
This expression shows that the target log-hazard ratio $\beta$ from the correctly specified model is distorted by the second term in \eqref{eqn;6.1}, which depends on the baseline cumulative hazard function, the true main effect parameter $\beta$, and the Laplace transform of the frailty distribution. We review and extend these bias factors for two popular univariate frailty models: the gamma and log-normal models.

\subsection{Gamma Frailty Model}

In addition to its mathematical convenience, Balan and Putter \cite{balan2020} noted that the phenomenon of a vanishing hazard ratio can be well explained by a univariate gamma frailty model, as will also be reviewed below. Abbring and van den Berg \cite{abbring2007} further justified the usefulness of the gamma frailty model by showing that, for a broad class of univariate frailty models, the frailty distribution among survivors converges to a gamma distribution as time tends to infinity under mild regularity conditions.

When the frailty term $W=\exp(\theta Z)$ follows a gamma distribution with mean 1 and variance $1/\kappa$, the corresponding Laplace transform is
\[
L(v)=\left(\frac{\kappa}{\kappa+v}\right)^{\kappa},
\]
which converges to $e^{-v}$ as $\kappa \rightarrow \infty$, equivalently as $\mathrm{Var}(W)\rightarrow 0$. Note that
\[
E\{\log(W)\}=E(\theta Z)=\Psi(\kappa)-\log(\kappa)
\]
and
\[
\mathrm{Var}\{\log(W)\}=\mathrm{Var}(\theta Z)=\Psi^{\prime}(\kappa),
\]
where $\Psi(\cdot)$ and $\Psi^{\prime}(\cdot)$ denote the digamma and trigamma functions, respectively. Therefore, standardizing the log-gamma random variable requires first subtracting its mean, $\Psi(\kappa)-\log(\kappa)$, and then dividing by $\sqrt{\Psi^{\prime}(\kappa)}$. When $\mathrm{Var}(Z)=1$, there is a direct relationship between the frailty parameter $\kappa$ and the omitted-covariate effect $\theta$:
\[
\theta^2=\Psi^{\prime}(\kappa),
\]
which does not have the solutions near $\theta \approx 0$.

Under Model III, the marginal survival function and the derivative of the Laplace transform are given by
\[
S^{(GAM)}(t;X)\equiv L(H_0(t)e^{\beta X})
=
\left(\frac{\kappa}{\kappa+H_0(t)e^{\beta X}}\right)^{\kappa},
\]
and
\[
L^{\prime}(H_0(t)e^{\beta X})
=
-
L(H_0(t)e^{\beta X})
\left(\frac{\kappa}{\kappa+H_0(t)e^{\beta X}}\right),
\]
respectively. Hence, the bias term in \eqref{eqn;6.1} simplifies to
\[
\mathrm{Bias}^{(GAM)}(t)
=
\log\left(\frac{\kappa+H_0(t)}{\kappa+H_0(t)e^{\beta}}\right).
\]
Note that this bias approaches $-\beta$ as $\kappa \rightarrow 0$ (i.e., $\mathrm{Var}(W)\rightarrow \infty$), and 0 as $\beta \rightarrow 0$ or $\kappa \rightarrow \infty$ (i.e., $\mathrm{Var}(W)\rightarrow 0$), provided that $\kappa \gg H_0(\tau)$, where $\tau$ denotes the largest event time. It is also straightforward to verify that the bias factor on the hazard-ratio scale is monotone increasing in $H_0(t)$ when $\beta>0$ and monotone decreasing in $H_0(t)$ when $\beta<0$.

\subsection{Log-normal Frailty Model}

The omitted covariate may correspond to a single biomarker measurement or to a linear combination of such measurements, such as gene expression levels or hormone receptor levels. Without loss of generality, we may assume that it has been normalized to follow a standard normal distribution. Balan and Putter \cite{balan2020} noted that, when matched on mean and variance, or equivalently on the coefficient of variation, the log-normal distribution is virtually indistinguishable from the inverse Gaussian distribution, which can explain a vanishing hazard ratio much like the gamma frailty distribution.

Suppose the omitted covariate $Z$ follows a standard normal distribution. Then $\log(W)=\theta Z$ follows a normal distribution with mean 0 and variance $\theta^2$, and $W$ follows a log-normal distribution with
\[
E(W)=\exp(\theta^2/2)
\]
and
\begin{equation}
\mathrm{Var}(W)=\exp(\theta^2)\{\exp(\theta^2)-1\}.
\label{eqn;6.11}
\end{equation}
Inverting (\ref{eqn;6.11}) gives
\begin{equation}
\theta^2=
\log\left(\frac{1+\sqrt{1+4\,\mathrm{Var}(W)}}{2}\right).
\label{eqn;6.12}
\end{equation}
A practical caution is that the variance of the log-normal frailty can increase rapidly as the omitted-covariate effect $\theta$ grows. Despite its usefulness in practice, the log-normal distribution does not have a closed-form Laplace transform, which hinders further analytical simplification of the bias formula.

\section{Existing Work on Bias Correction under Model Misspecification}

Struthers and Kalbfleisch \cite{struthers1986} investigated the properties of the estimator under misspecification of the proportional hazards model, for example when covariates are omitted. Henderson and Oman \cite{henderson1999} modified the partial likelihood estimating equation under random censoring by modeling the exponentiated omitted covariate as a univariate frailty term and approximating the solution near the null value of the main effect. Under the additional assumption that all covariates are centered and in the absence of censoring, the bias factor simplifies to $\nu=\kappa/(1+\kappa)$ under the gamma frailty model with mean 1 and variance $1/\kappa$. Furthermore, under the Koziol--Green (KG) censoring model \cite{koziol1976}, in which the proportional hazards assumption holds between the baseline event-time and censoring-time distributions, the bias factor further simplifies to
\begin{equation}
\nu=
\frac{\kappa \alpha_0^2 E_{\kappa}(\kappa \alpha_0)+(1-\alpha_0)e^{-\kappa \alpha_0}}
     {(\kappa+1)E_{\kappa+1}(\kappa \alpha_0)},
\label{eqn;4.0}
\end{equation}
where
\[
E_a(b)=\int_1^{\infty}\frac{\exp(-bv)}{v^a}\,dv
\]
and $\alpha_0$ is the proportionality parameter, which can be expressed as a function of the baseline censoring proportion $p_0$ as
$\alpha_0=p_0/(1-p_0)$. Equation (\ref{eqn;4.0}) may be useful in practice when the true values of $\alpha_0$ and $\kappa$ are known.

\section{Does Univariate Frailty Analysis Help Recover the Target Hazard Ratio?}

The observation that treatment effect estimates from the two most commonly used univariate frailty models, gamma and log-normal, can differ substantially from one another, even though both are often closer to the target value than the estimate from a model that completely ignores the omitted covariate, raises a practical question: to what extent can univariate frailty analysis help recover the target hazard ratio, and how does its performance depend on the underlying frailty distribution?

To address this question, we conducted a simulation study under independent uniform random censoring, a setting frequently encountered in clinical trial data. A group indicator $X$ was generated from a Bernoulli distribution with success probability 0.5, and the frailty term $W=\exp(\theta Z)$ was generated from either a gamma distribution with mean 1 and variance $1/\kappa$, or a log-normal distribution such that $\log(W)$ had mean 0 and variance $\theta^2$. For a fair comparison, the variances of the two frailty distributions were matched using
\[
\theta_{LN}=\sqrt{\log\left\{\frac{1+\sqrt{1+4/\kappa}}{2}\right\}},
\]
where $\theta_{LN}$ is the omitted covariate effect, or equivalently the standard deviation on the log scale, under the log-normal frailty model, and $\kappa$ satisfies $\Psi^{\prime}(\kappa)=\theta_{GAM}$, where $\theta_{GAM}$ denotes the omitted covariate effect under the gamma frailty model.

Under the Cox proportional hazards model \cite{cox1972,cox1976} with a Weibull baseline hazard having scale parameter $\sigma$ and shape parameter $\rho$, the true event time was generated using the probability integral transform:
\[
T=\left[\frac{-\log(1-V)}{\sigma \exp(\beta X)W}\right]^{1/\rho},
\]
where $V\sim U(0,1)$. The censoring time was generated independently from a uniform distribution, $C\sim U(0,a)$, where $a$ controls the censoring proportion. The observed survival time was then defined as $Y=\min(T,C)$.

The baseline Weibull parameters were set to $(\sigma,\rho)=(0.2,2.0)$ or $(0.4,1.0)$. Supplementary Figure S1 shows the two baseline Weibull distributions used in the simulation study: a right-skewed distribution with intermediate follow-up and an exponential distribution with longer follow-up. For values of the target hazard ratio $\exp(\beta)$ equal to 0.3, 0.5, 0.8, and 1.0, and for omitted covariate effects under the gamma frailty model of $\theta_{GAM}=1.0$, 1.5, and 2.0, we compared treatment effect estimates across several analyses. The values of $\theta_{GAM}$ correspond to $\theta_{LN}=0.624$, 0.735, and 0.808 under the log-normal frailty model, yielding matched frailty variances of 0.701, 1.232, and 1.766, respectively.

Specifically, we compared the estimated treatment effect parameter under the following settings:
\begin{enumerate}
    \item the correctly specified proportional hazards model including the omitted covariate $Z$, when the underlying frailty distribution was gamma or log-normal, denoted by $\hat \beta^{(adj)}_{GAM}$ and $\hat \beta^{(adj)}_{LN}$;
    \item the proportional hazards model omitting $Z$, when the underlying frailty distribution was gamma or log-normal, denoted by $\hat \beta^{(omit)}_{GAM}$ and $\hat \beta^{(omit)}_{LN}$;
    \item the univariate gamma frailty analysis, when the underlying frailty distribution was correctly specified as gamma or misspecified as log-normal, denoted by $\hat \beta^{(PH)}_{GAM}$ and $\hat \beta^{(ME)}_{GAM}$; and
    \item the univariate log-normal frailty analysis, when the underlying frailty distribution was correctly specified as log-normal or misspecified as gamma, denoted by $\hat \beta^{(ME)}_{LN}$ and $\hat \beta^{(PH)}_{LN}$.
\end{enumerate}
These comparisons were used to assess how closely each analysis recovered the target treatment effect. We also included the analytic correction formula in (\ref{eqn;4.0}) when the true frailty distribution was gamma, denoted by $\hat \beta^{(HO)}_{GAM}$, to evaluate its empirical performance under practical scenarios, although this comparison is not entirely parallel because it relies on different assumptions.

We assumed uniform censoring on $(0,10)$. The sample size was 500 per group, and each scenario was replicated 2,000 times. Results were similar for smaller and larger sample sizes (not shown). The censoring proportion ($CP_{GAM}$ or $CP_{LN}$) depended on the target treatment effect $\beta$, the omitted covariate effect $\theta$, the Weibull parameters $\sigma$ and $\rho$, and the assumed frailty distribution.

Table \ref{table_1} summarizes the simulation results in terms of the empirical medians and interquartile ranges (IQRs) of the estimates. First, when the omitted covariate was included in the proportional hazards model, the treatment effect estimates were, as expected, close on average to the true parameter values under both frailty distributions. Second, when the exponentiated omitted covariate was assumed to follow a gamma distribution, the estimates from the gamma frailty analysis were noticeably closer to the true parameter values than those obtained either by completely ignoring the omitted covariate or by misspecifying the frailty distribution as log-normal. When the omitted covariate followed a standard normal distribution, corresponding to a log-normal frailty model, misspecification of the frailty distribution as gamma did not increase the bias substantially relative to the analysis that ignored the omitted covariate, although the bias tended to increase as the omitted covariate effect became larger. Overall, these findings suggest that assuming a log-normal distribution for the exponentiated omitted covariate may be the more practically robust choice for recovering the target treatment effect on average across both gamma and log-normal frailty analyses. The analytic correction of Henderson and Oman tended to overestimate the treatment effect, particularly when the true treatment effect increases, as expected, and the omitted covariate effect was small. Across the scenarios considered, the censoring proportion ranged from approximately 28\% to 63\% under the gamma frailty model and from approximately 20\% to 42\% under the log-normal frailty model.

\section{A Simulation-Based Graphical Approach to Recover the Target Hazard Ratio}

The analytic expressions presented in Section 4 are useful for understanding the asymptotic behavior of the bias factors. In practice, however, only a single data set is observed from an unknown population. Any attempt to recover the target parameter values analytically from that observed data set may therefore be sensitive to sampling variability in addition to the bias induced by omission of the covariate, especially when the sampling variability of the parameter estimates is large, as can occur for $\kappa$ under the gamma frailty model.

To address this limitation, we propose a simulation-based graphical approach that does not rely on parameter estimates obtained from the observed data or on direct use of frailty-model theory for bias correction. The main idea is that, under the minimal assumptions of a Weibull baseline event-time distribution and unit variance of the unobservable omitted covariate, the target population parameters are those that would generate a band of survival curves containing the observed survival curves with maximal empirical coverage. We refer to this as the Tightest Survival Band (TSB) — tightest in the sense of the narrowest band given the observed sample size. The optimization criterion is based on the simple observation that a function $f(x)$ is contained between two functions $g(x)$ and $h(x)$ if and only if
\[
\{f(x)-g(x)\}\{f(x)-h(x)\}\le 0
\]
for all $x$. We now describe the TSB procedure.

\subsection{Tightest Survival Band (TSB)}

In Section 5, we observed that univariate frailty analysis may help recover the target treatment effect to some extent \emph{on average} under a range of scenarios, particularly under the log-normal frailty model. To further improve bias reduction for a single observed sample, we propose the following simulation-based TSB procedure.

\begin{enumerate}
\item Using the \emph{observed data}, compute the Kaplan--Meier estimates stratified by the group indicator $X$, and evaluate them on a common grid of time points $g_i$ $(i=1,2,\ldots,n_g)$. Denote these estimates by $\hat S^{(obs)}_x(g_i)$ for $x=0,1$.
\item For candidate parameter values of $\beta$, $\theta$, $\sigma$, and $\rho$, repeat the following steps $N$ times $(j=1,2,\ldots,N)$:
    \begin{enumerate}
        \item Determine $\kappa$ from $\Psi'(\kappa)=\theta^2$ under the assumption $\mathrm{Var}(Z)=1$ if a gamma frailty distribution is used for the unobservable omitted covariate.

\begin{landscape}
\renewcommand{\arraystretch}{1.1}

\begin{table}[p]
	\centering
	\caption{Comparison of treatment effect estimates, summarized by empirical medians and interquartile ranges (IQRs), from adjusted, unadjusted, gamma frailty, and log-normal frailty models under correct and incorrect specification of the underlying frailty distribution.}
	\label{table_1}
	\resizebox{1.6\textwidth}{!}{%
		\begin{threeparttable}
\begin{tabular}{*{16}{r}}
\toprule
$\sigma$ & $\rho$ & $\beta$ & $HR$ & $\theta_{GAM}$ & $CP_{GAM}$ & $\hat \beta^{(adj)}_{GAM}$ &
$\hat \beta^{(omit)}_{GAM}$ & $\hat \beta^{(PH)}_{GAM}$ & $\hat \beta^{(ME)}_{GAM}$ &
$\hat \beta^{(HO)}_{GAM}$ & $CP_{LN}$ & $\hat \beta^{(adj)}_{LN}$ &
$\hat \beta^{(omit)}_{LN}$ & $\hat \beta^{(PH)}_{LN}$ & $\hat \beta^{(ME)}_{LN}$ \\
\midrule
0.20 & 2 & -1.204 & 0.3 & 1.0 & 0.360 &
-1.203 (0.114) & -0.830 (0.114) & -1.158 (0.215) &
-0.987 (0.144) & -1.840 (0.896) & 0.285 &
-1.207 (0.110) & -1.099 (0.111) & -1.142 (0.153) & -1.169 (0.137) \\

 &  &  &  & 1.5 & 0.466 &
-1.205 (0.127) & -0.579 (0.118) & -1.046 (0.300) &
-0.669 (0.146) & -1.616 (2.497) & 0.289 &
-1.206 (0.110) & -1.024 (0.111) & -1.109 (0.181) & -1.121 (0.136) \\

 &  &  &  & 2.0 & 0.551 &
-1.205 (0.138) & -0.436 (0.128) & -0.870 (0.359) &
-0.476 (0.142) & -1.215 (1.756) & 0.293 &
-1.206 (0.110) & -0.970 (0.114) & -1.087 (0.196) & -1.080 (0.134) \\ \cline{3-16}

 &  & -0.693 & 0.5 & 1.0 &  0.317 &
-0.691 (0.104) & -0.472 (0.107) & -0.621 (0.202) &
-0.535 (0.126) & -0.953 (0.277) & 0.244 &
-0.696 (0.102) & -0.630 (0.104) & -0.650 (0.119) & -0.665 (0.112) \\

 &  &  &  & 1.5 & 0.431 &
-0.694 (0.118) & -0.325 (0.113) & -0.512 (0.290) &
-0.355 (0.127) & -0.828 (0.884) & 0.248 &
-0.697 (0.103) & -0.585 (0.103) & -0.620 (0.134) & -0.628 (0.115) \\

 &  &  &  & 2.0 & 0.525 &
-0.693 (0.128) & -0.242 (0.121) & -0.368 (0.301) &
-0.255 (0.132) & -0.592 (0.659) & 0.252 &
-0.695 (0.100) & -0.553 (0.103) & -0.600 (0.146) & -0.599 (0.119) \\ \cline{3-16}

 &  & -0.223 & 0.8 & 1.0 & 0.282 &
-0.222 (0.098) & -0.153 (0.104) & -0.157 (0.116) &
-0.161 (0.110) & -0.245 (0.126) & 0.214 &
-0.225 (0.098) & -0.201 (0.095) & -0.203 (0.100) & -0.208 (0.099) \\

 &  &  &  & 1.5 & 0.401 &
-0.225 (0.108) & -0.102 (0.108) & -0.106 (0.123) &
-0.104 (0.113) & -0.212 (0.142) & 0.218 &
-0.225 (0.097) & -0.185 (0.098) & -0.187 (0.102) & -0.193 (0.102) \\

 &  &  &  & 2.0 & 0.501 &
-0.221 (0.124) & -0.074 (0.118) & -0.075 (0.125) &
-0.074 (0.120) & -0.173 (0.131) & 0.221 &
-0.224 (0.095) & -0.175 (0.099) & -0.177 (0.102) & -0.181 (0.102) \\ \cline{3-16}

 &  & 0.000 & 1.0 & 1.0 & 0.268 &
0.003 (0.100) & -0.004 (0.103) & -0.004 (0.103) &
-0.004 (0.104) & -0.057 (0.152) & 0.202 &
0.000 (0.100) & 0.002 (0.097) & 0.002 (0.097) & 0.002 (0.098) \\

 &  &  &  & 1.5 & 0.387 &
-0.002 (0.106) & 0.001 (0.108) & 0.001 (0.108) &
0.001 (0.108) & 0.043 (0.208) & 0.206 &
-0.001 (0.097) & 0.001 (0.098) & 0.001 (0.098) & 0.001 (0.100) \\

 &  &  &  & 2.0 & 0.491 &
0.002 (0.122) & 0.001 (0.116) & 0.001 (0.116) &
0.001 (0.116) & 0.041 (0.203) & 0.209 &
-0.001 (0.097) & 0.002 (0.097) & 0.002 (0.097) & 0.003 (0.099) \\ \hline

0.40 & 1 & -1.204 & 0.3 & 1.0 & 0.501 &
-1.203 (0.128) & -0.942 (0.126) & -1.155 (0.254) &
-1.022 (0.143) & -2.002 (1.122) & 0.417 &
-1.207 (0.119) & -1.123 (0.118) & -1.166 (0.160) & -1.175 (0.136) \\

 &  &  &  & 1.5 & 0.576 &
-1.205 (0.138) & -0.687 (0.136) & -1.074 (0.334) &
-0.747 (0.151) & -1.822 (2.382) & 0.419 &
-1.207 (0.119) & -1.065 (0.117) & -1.136 (0.194) & -1.128 (0.134) \\

 &  &  &  & 2.0 & 0.628 &
-1.210 (0.148) & -0.508 (0.139) & -0.914 (0.391) &
-0.540 (0.149) & -2.011 (1.726) & 0.421 &
-1.207 (0.118) & -1.020 (0.117) & -1.115 (0.208) & -1.090 (0.136) \\ \cline{3-16}

 &  & -0.693 & 0.5 & 1.0 & 0.442 &
-0.690 (0.115) & -0.529 (0.115) & -0.635 (0.204) &
-0.564 (0.125) & -1.018 (0.447) & 0.345 &
-0.696 (0.107) & -0.641 (0.104) & -0.670 (0.132) & -0.669 (0.117) \\

 &  &  &  & 1.5 & 0.535 &
-0.694 (0.127) & -0.380 (0.129) & -0.545 (0.300) &
-0.401 (0.137) & -0.890 (0.853) & 0.351 &
-0.695 (0.104) & -0.605 (0.107) & -0.645 (0.146) & -0.635 (0.117) \\

 &  &  &  & 2.0 & 0.599 &
-0.693 (0.139) & -0.278 (0.132) & -0.411 (0.313) &
-0.289 (0.139) & -0.696 (0.692) & 0.355 &
-0.695 (0.103) & -0.578 (0.109) & -0.629 (0.157) & -0.608 (0.116) \\ \cline{3-16}

 &  & -0.223 & 0.8 & 1.0 & 0.388 &
-0.221 (0.109) & -0.170 (0.110) & -0.178 (0.128) &
-0.174 (0.115) & -0.284 (0.177) & 0.282 &
-0.224 (0.100) & -0.204 (0.101) & -0.209 (0.108) & -0.210 (0.104) \\

 &  &  &  & 1.5 & 0.497 &
-0.225 (0.118) & -0.116 (0.118) & -0.123 (0.142) &
-0.117 (0.123) & -0.235 (0.153) & 0.290 &
-0.224 (0.098) & -0.192 (0.100) & -0.197 (0.111) & -0.197 (0.105) \\

 &  &  &  & 2.0 & 0.572 &
-0.223 (0.130) & -0.087 (0.125) & -0.088 (0.141) &
-0.087 (0.127) & -0.189 (0.159) & 0.296 &
-0.225 (0.096) & -0.183 (0.101) & -0.189 (0.109) & -0.189 (0.105) \\ \cline{3-16}

 &  & 0.000 & 1.0 & 1.0 & 0.363 &
0.003 (0.106) & -0.006 (0.109) & -0.006 (0.110) &
-0.006 (0.110) & 0.045 (0.179) & 0.256 &
-0.001 (0.095) & 0.000 (0.096) & 0.000 (0.097) & 0.000 (0.098) \\

 &  &  &  & 1.5 & 0.479 &
-0.002 (0.118) & 0.003 (0.117) & 0.003 (0.119) &
0.003 (0.118) & 0.061 (0.170) & 0.264 &
-0.001 (0.096) & 0.002 (0.099) & 0.002 (0.099) & 0.002 (0.100) \\

 &  &  &  & 2.0 & 0.560 &
0.000 (0.126) & 0.000 (0.126) & 0.000 (0.126) &
0.000 (0.126) & -0.046 (0.195) & 0.272 &
-0.002 (0.097) & 0.001 (0.096) & 0.001 (0.096) & 0.001 (0.097) \\

\bottomrule
\end{tabular}%

\begin{tablenotes}
\item \underline{Definitions of column names}: $\beta$ and $HR$ denote the true treatment effect parameter and the corresponding hazard ratio ($HR=\exp(\beta)$), respectively; $\theta_{GAM}$ denotes the omitted covariate effect under the gamma frailty model; $CP_{GAM}$ denotes the censoring proportion under the gamma frailty model; $\hat \beta^{(adj)}_{GAM}$ denotes the estimated treatment effect when the omitted covariate $Z$ is included and the true frailty distribution is gamma; $\hat \beta^{(omit)}_{GAM}$ denotes the estimated treatment effect when $Z$ is omitted and the true frailty distribution is gamma; $\hat \beta^{(PH)}_{GAM}$ denotes the estimated treatment effect under the correctly specified gamma frailty model when the true frailty distribution is gamma; $\hat \beta^{(ME)}_{GAM}$ denotes the estimated treatment effect under the misspecified log-normal frailty model when the true frailty distribution is gamma; $\hat \beta^{(HO)}_{GAM}$ denotes the estimated treatment effect from the analytic correction formula \eqref{eqn;4.0} of Henderson and Oman when the true frailty distribution is gamma; $CP_{LN}$ denotes the censoring proportion under the log-normal frailty model; $\hat \beta^{(adj)}_{LN}$ denotes the estimated treatment effect when the omitted covariate $Z$ is included and the true frailty distribution is log-normal; $\hat \beta^{(omit)}_{LN}$ denotes the estimated treatment effect when $Z$ is omitted and the true frailty distribution is log-normal; $\hat \beta^{(PH)}_{LN}$ denotes the estimated treatment effect under the misspecified gamma frailty model when the true frailty distribution is log-normal; and $\hat \beta^{(ME)}_{LN}$ denotes the estimated treatment effect under the correctly specified log-normal frailty model when the true frailty distribution is log-normal.
\end{tablenotes}
		\end{threeparttable}
	}
\end{table}
\end{landscape}

        \item For a sample size $n$ equal to that of the observed data, generate a treatment indicator $X_{jl}$ $(l=1,2,\ldots,n)$ from a Bernoulli distribution with success probability 0.5, and generate a frailty term $W_{jl}$ from either:
        \begin{itemize}
            \item a gamma distribution with mean 1 and variance $1/\kappa$, or
            \item a log-normal distribution such that $\log(W_{jl})$ has mean 0 and variance $\theta^2$.
        \end{itemize}

        \item Under the Weibull baseline assumption, generate event times from
        \[
        T_{jl}
        =
        \left[
        \frac{-\log(1-V_{jl})}
        {\sigma \exp(\beta X_{jl})W_{jl}}
        \right]^{1/\rho},
        \qquad l=1,2,\ldots,n,
        \]
        where $V_{jl}\sim U(0,1)$.

        \item Generate censoring times from a mixture of two uniform distributions,
        \[
        C_{jl} \sim \omega_1 U(0,a_1) + (1-\omega_1)U(a_1,a_2),
        \qquad l=1,2,\ldots,n,
        \]
        where $0\le \omega_1 \le 1$ is the proportion of censored observations occurring during the first interval $[0,a_1]$, and $a_2$ is the maximum follow-up time. As a special case, if censoring is assumed to be uniform over the entire study period, say $C_{jl}\sim U(0,a)$, then one may set $\omega_1=1$ and $a_1=a$.

        \item Define the observed survival time and event indicator by
        \[
        Y_{jl}=\min(T_{jl},C_{jl})
        \quad\text{and}\quad
        \Delta_{jl}=I(T_{jl}\le C_{jl}).
        \]

        \item Using this \emph{simulated data set}, compute the Kaplan--Meier estimates stratified by $X$, and evaluate them on the same common grid $g_i$ $(i=1,2,\ldots,n_g)$. Denote these estimates by $\hat S^{(sim)}_{x,j}(g_i)$ for $x=0,1$ and $j=1,2,\ldots,N$.
    \end{enumerate}

\item For each group $x$ and each grid point $g_i$, compute the minimum and maximum values of the simulated survival estimates across the $N$ simulated data sets:
\[
LO_x(g_i)=\min_{j=1,\ldots,N}\left\{\hat S^{(sim)}_{x,j}(g_i)\right\},
\qquad
UP_x(g_i)=\max_{j=1,\ldots,N}\left\{\hat S^{(sim)}_{x,j}(g_i)\right\}.
\]
The interval
\[
SPSR_x(g_i)=[LO_x(g_i),\,UP_x(g_i)]
\]
will be referred to as the Simulated Pointwise Survival Range (SPSR) for group $x$ at time point $g_i$.

\item For each group $x$, calculate the number of grid points at which the observed survival estimates are contained in the SPSR:
\[
f_x
=
\sum_{i=1}^{n_g}
I\left(
\left\{\hat S^{(obs)}_x(g_i)-LO_x(g_i)\right\}
\left\{\hat S^{(obs)}_x(g_i)-UP_x(g_i)\right\}
\le 0
\right).
\]

\item Define the overall proportion of grid points at which the observed survival estimates are contained in the SPSR across both groups as the \emph{coverage coefficient}:
\[
CC=\frac{\sum_{x=0}^1 f_x}{2n_g}.
\]

\item Using a grid search, identify the values of $\beta$, $\theta$, $\sigma$, and $\rho$ that maximize $CC$. The band consisting of SPSRs over all grid points that gives the maximal $CC$ will be referred to as Tightest Survival Band (TSB). If the maximizer is not unique and the surface of $CC$ is flat over a range of parameter values, then use the midpoint of the minimum and maximum values in that range, or equivalently the median of the maximizing set, as a practical summary.
\end{enumerate}

The simulation results presented later suggest that the choice of frailty distribution has only a minimal effect on recovery of the target treatment effect. In Appendix A, we provide a formal justification for this empirical behavior by establishing that, under standard identifiability conditions and almost sure convergence of the pointwise extrema of simulated Kaplan–Meier curves to their true sampling range, the TSB grid maximizer concentrates around the true parameter vector, at the sample size of the observed data, as the number of simulation replicates increases.

\subsection{A Numerical Example}

For numerical illustration, we first assumed that the frailty term, i.e., the exponentiated omitted covariate, followed a gamma distribution with mean 1 and variance $1/\kappa$, where $\kappa=0.57$ was obtained from the equation $\Psi^{\prime}(\kappa)=\theta^2$ with the true omitted covariate effect set to $\theta=2$. This implies a frailty variance of $1/\kappa=1.75$. We chose $(\sigma,\rho)=(0.2,2.0)$ for the baseline Weibull distribution, corresponding to a unimodal distribution skewed to the right with intermediate follow-up, and set the true hazard ratio to $\exp(\beta)=0.3$, that is, $\beta=-1.20$. The sample size was set to 1,500 per group to mimic the NSABP B-14 data. Under this scenario, the overall censoring proportion was approximately 56.3\%.

Using a fixed random seed, we generated one data set from this population. When the omitted covariate $Z$ was completely ignored, the naive estimate of the treatment effect parameter was $\hat\beta=-0.40$. A univariate gamma frailty analysis produced $\hat\beta=-1.23$, which was much closer to the true value. The estimated gamma frailty parameter was $\hat\kappa=0.19$, corresponding to an estimated omitted covariate effect of $\hat\theta=5.13$.

Next, under the same Weibull baseline setting, we generated another data set in which $\log(W)$ was normally distributed, corresponding to a log-normal frailty model. In this case, the true omitted covariate effect was set to $\theta=0.8$ so that the frailty variance would be similar to that under the gamma frailty setting. A univariate log-normal frailty analysis yielded $\hat\beta=-1.05$ for the treatment effect and $\hat\theta=0.42$ for the omitted covariate effect. These two realized data sets may be regarded as the \emph{observed} data sets in practice and will be referred to as such below.

It is natural to begin the grid search for the TSB procedure using parameter estimates obtained from the observed data as initial values. Figure \ref{fig:numeric} shows the observed survival curves (red solid lines) together with the overlaid TSBs based on 1,000 simulated replicates (gray dotted lines). Figures \ref{fig:numeric}(a) and (b) correspond to the gamma frailty case, using the parameter estimates from the observed data and the true parameter values, respectively. Figures \ref{fig:numeric}(c) and (d) show the analogous results for the log-normal frailty case.
\begin{center}
\begin{figure}[!htbp]
\center
    \includegraphics[width = 1.0\linewidth,keepaspectratio]{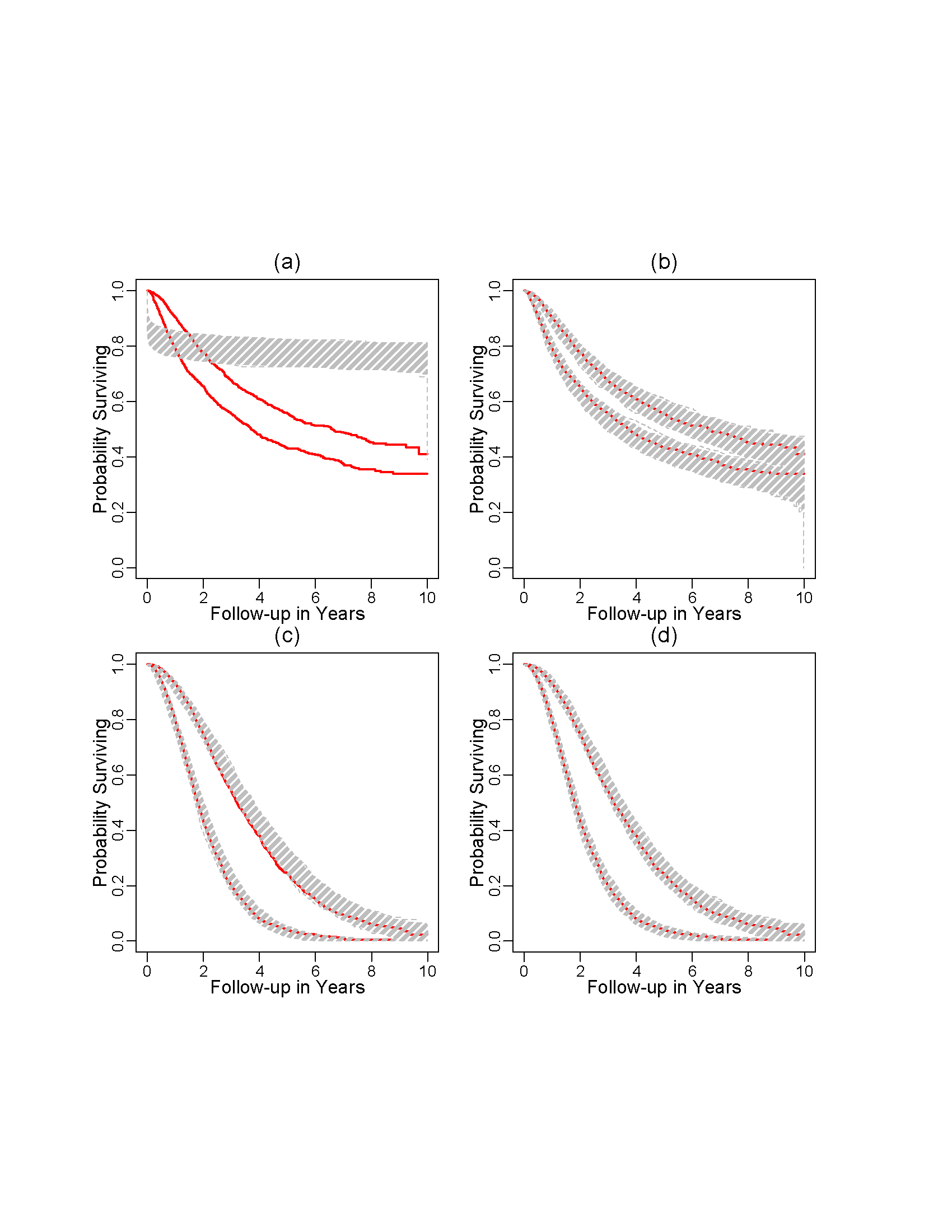}
    \vspace{-1.0in}
    \caption{A numerical illustration of the TSB method: (a) using parameter estimates from the univariate gamma frailty analysis; (b) using the true parameter values under the gamma frailty model; (c) using parameter estimates from the univariate log-normal frailty analysis; and (d) using the true parameter values under the log-normal frailty model.}
    \label{fig:numeric}
\end{figure}
\end{center}

For the gamma frailty setting, the observed survival curves are barely contained within the TSB when the parameter estimates from the observed data are used (Figure \ref{fig:numeric}(a)), whereas they are fully contained when the true parameter values are used (Figure \ref{fig:numeric}(b)), as desired. In fact, the corresponding $CC$ values were 0.06 and 1.00, respectively. For the log-normal frailty setting, the TSB contains the observed survival curves reasonably well in both cases, with $CC$ values of 0.80 and 0.92 when using the estimated and true parameter values, respectively. Overall, regardless of the assumed frailty distribution, the TSB procedure captures the observed survival curves well when evaluated at the true parameter values.

\subsection{Impact of Treatment Effect and Sample Size on the Coverage Coefficient ($CC$)}

In the numerical example above, the sample size was as large as 1,500 per group to mimic the NSABP B-14 data, and the treatment effect was also strong, with a true hazard ratio of 0.3, making the difference between the two groups readily visible. It is therefore of interest to examine how the proposed approach performs across a wider range of sample sizes and treatment-effect magnitudes, since the variability of the simulated survival curves generated by the TSB procedure is expected to increase as the sample size decreases.

\subsubsection{When $\theta$, $\sigma$, and $\rho$ are fixed}

We first conducted a simulation study to assess the behavior of the optimization function ($CC$) as a function of the treatment effect parameter $\beta$, conditional on the true values of the remaining parameters $\theta$, $\sigma$, and $\rho$. Simulations were performed with sample sizes per group of 100, 250, 500, and 1,000, and with true hazard ratios of 0.3, 0.5, and 0.8 under the two Weibull baseline scenarios described earlier.

Under the gamma frailty setting, the omitted covariate effect was fixed at $\theta=2.0$, corresponding to $\kappa=0.57$ and hence $\mathrm{Var}(W)=1.75$. The results in Supplementary Table S1 show that the surface of the optimization function, quantified by $CC$, tends to be flat near its maximum, and that the width of the flat interval, $\exp(\hat\beta_U)-\exp(\hat\beta_L)$, increases as the sample size decreases. For example, when $\sigma=0.4$, $\rho=1.0$, and $\exp(\beta)=0.5$, the flat intervals were [0.36, 0.62], [0.35, 0.76], [0.28, 0.89], and [0.14, 1.05] as $n$ decreased from 1,000 to 500 to 250 to 100, while $\max(CC)\approx 1$ in all cases.

If the recovered hazard ratio is summarized by the midpoint of the flat interval,
\[
\exp(\bar{\hat\beta})=\{\exp(\hat\beta_U)+\exp(\hat\beta_L)\}/2,
\]
then the recovered treatment effect tends to approach the true value as the sample size increases, although some fluctuation remains because of sampling variability. In particular, when $n=1{,}000$, the true hazard ratios were recovered accurately under both Weibull baseline scenarios. Supplementary Table S1 presents the corresponding results for the log-normal frailty setting. Taken together, Table 2 and Supplementary Table S1 indicate that, regardless of the underlying frailty distribution, the proposed TSB method recovers the target hazard ratio reasonably well when the other parameters are fixed, especially for moderate to large sample sizes. This empirical pattern, in which the flat maximum region narrows as $n$ increases, is consistent with the asymptotic argument given in Appendix A.

\begin{table}[!htbp]
	\centering
	\caption{Effects of treatment effect size and sample size on recovery of the target hazard ratio when the underlying frailty distribution is gamma, conditional on the omitted covariate effect and baseline Weibull parameters}
	\label{table2}
	\resizebox{1.0\textwidth}{!}{%
		\begin{threeparttable}
\begin{tabular}{l|ccccccc}
\hline\hline
\multicolumn{1}{c|}{}
  & Sample Size ($n$) & $\beta$ & $HR=\exp(\beta)$ & max($CC$) & $\exp(\hat \beta_L)$ & $\exp(\hat \beta_U)$ & $\exp(\bar{\hat \beta})$ \\  \hline
\multirow{9}{*}{$\sigma=0.2,\rho=2.0$} 
&100 & -1.204 & 0.3 & 1.000 & 0.09 & 0.79 & 0.440 \\
& & -0.693 & 0.5 & 1.000 & 0.14 & 0.86 & 0.500 \\
& & -0.223 & 0.8 & 1.000 & 0.23 & 0.99 & 0.610 \\
\cline{3-8}
&250 & -1.204 & 0.3 & 0.985 & 0.17 & 0.50 & 0.335 \\
& & -0.693 & 0.5 & 0.985 & 0.27 & 0.84 & 0.555 \\
& & -0.223 & 0.8 & 0.985 & 0.44 & 1.29 & 0.865 \\
\cline{3-8}
&500 & -1.204 & 0.3 & 1.000 & 0.21 & 0.41 & 0.310 \\
& & -0.693 & 0.5 & 1.000 & 0.34 & 0.75 & 0.545 \\
& & -0.223 & 0.8 & 1.000 & 0.84 & 1.18 & 1.010 \\
\cline{3-8}
&1000 & -1.204 & 0.3 & 1.000 & 0.22 & 0.37 & 0.295 \\
& & -0.693 & 0.5 & 1.000 & 0.36 & 0.66 & 0.510 \\
& & -0.223 & 0.8 & 1.000 & 0.58 & 1.01 & 0.795 \\
\hline
\multirow{9}{*}{$\sigma=0.4, \rho=1.0$} 
&100 & -1.204 & 0.3 & 1.000 & 0.09 & 0.64 & 0.365 \\
& & -0.693 & 0.5 & 1.000 & 0.14 & 1.05 & 0.595 \\
& & -0.223 & 0.8 & 1.000 & 0.21 & 1.60 & 0.905 \\
\cline{3-8}
&250 & -1.204 & 0.3 & 0.985 & 0.17 & 0.51 & 0.340 \\
& & -0.693 & 0.5 & 0.985 & 0.28 & 0.89 & 0.585 \\
& & -0.223 & 0.8 & 0.985 & 0.46 & 1.20 & 0.830 \\
\cline{3-8}
&500 & -1.204 & 0.3 & 1.000 & 0.22 & 0.42 & 0.320 \\
& & -0.693 & 0.5 & 1.000 & 0.35 & 0.76 & 0.555 \\
& & -0.223 & 0.8 & 1.000 & 0.58 & 1.14 & 0.860 \\ 
\cline{3-8}
&1000 & -1.204 & 0.3 & 1.000 & 0.22 & 0.38 & 0.300 \\
& & -0.693 & 0.5 & 1.000 & 0.36 & 0.62 & 0.490 \\
& & -0.223 & 0.8 & 1.000 & 0.60 & 1.01 & 0.805 \\ 
      \hline \hline
\end{tabular}
\begin{tablenotes}
\item \underline{Definitions of column names}: $\beta$ and $HR$ denote the true treatment effect parameter and its associated hazard ratio ($HR=\exp(\beta)$), respectively; max($CC$) is the maximum value of the coverage coefficient; $\exp(\hat \beta_L)$ and $\exp(\hat \beta_U)$ are the minimum and maximum values of the recovered hazard ratio attaining max($CC$); and $\exp(\bar{\hat \beta})$ is the midpoint of that interval.
\end{tablenotes}
		\end{threeparttable}
	}
\end{table}

\subsubsection{When $\sigma$ and $\rho$ are fixed}

To further evaluate the optimization function, we constructed contour plots of the $CC$ surface as a function of both the treatment effect and the omitted covariate effect, conditional on the true baseline Weibull parameters. The observed data were generated under true hazard ratios of 0.3, 0.5, and 0.8, with the true omitted covariate effect fixed at $\theta=2$ for the gamma frailty setting and $\theta=0.8$ for the log-normal frailty setting.

The TSB method was then applied over broad grid-search ranges chosen to contain the true values. The sample size per group was again set to 1,500, as in the B-14 data set. Figure 3 and Supplementary Figure S2 show the resulting contour plots under gamma and log-normal frailty settings, respectively. In each figure, the first column corresponds to $(\sigma,\rho)=(0.2,2.0)$ and the second column to $(\sigma,\rho)=(0.4,1.0)$. The red dashed lines indicate the true values of the treatment effect and omitted covariate effect, and the blue dotted lines indicate the corresponding recovered values obtained from the TSB method as marginal medians over the flat maximum surface.

For both frailty settings, the optimization function $CC$ exhibits a flat maximum region that includes the true parameter values, conditional on the baseline Weibull parameters. The recovered values are close to the true values, particularly in the gamma frailty setting considered here. As discussed in Appendix A, the sharpness versus flatness of this maximum region reflects how well the parameters are identified locally by the observed survival curves.

\begin{center}
\begin{figure}[!htbp]
\center
    \includegraphics[width = 0.8\linewidth,keepaspectratio]{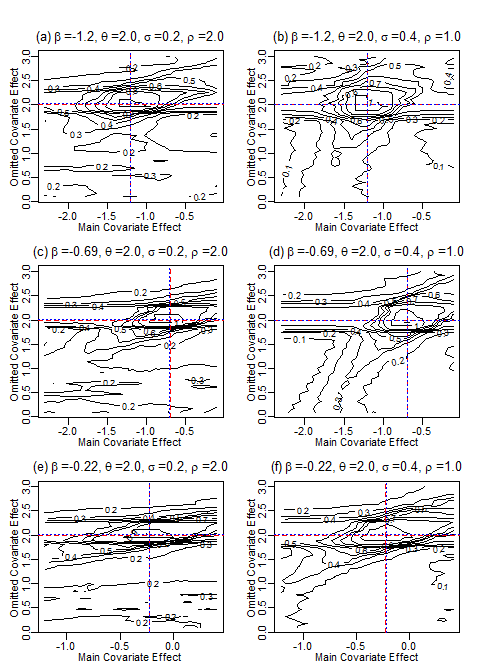}
    \vspace{0.0in}
    \caption{Contour plots of the optimization function ($CC$) under a gamma frailty setting. The true treatment effects correspond to hazard ratios of 0.3, 0.5, and 0.8, with omitted covariate effect fixed at $\theta=2$ and baseline Weibull parameters $(\sigma,\rho)=(0.2,2.0)$ or $(0.4,1.0)$. The red dashed lines indicate the true values of the treatment effect and omitted covariate effect, and the blue dotted lines indicate the corresponding recovered values obtained from the TSB method.}
    \label{fig:contour_gamma}
\end{figure}
\end{center}

\subsubsection{Generalization to Repeated Observed Data Sets When $\sigma$ and $\rho$ are Fixed}

Section 6.3.2 considered one observed data set at a time under specific scenarios. To evaluate the performance of the proposed method more broadly, we repeated the TSB procedure over 1,000 independently generated observed data sets under log-normal frailty settings with $\theta_{LN}=0.8$ (equivalently, $\theta_{GAM}\approx 2.0$), true hazard ratios of 0.3, 0.5, and 0.8, and sample sizes of 500, 1,000, and 1,500 per group.

Table \ref{table_4} summarizes the medians and interquartile ranges of the recovered treatment effect and omitted covariate effect. Overall, the proposed TSB method recovers both parameters reasonably well, with the true values generally lying within the interquartile ranges and performance improving as the sample size increases.

\begin{table}[!htbp]
	\centering
	\caption{Medians (interquartile ranges) of recovered treatment-effect and omitted-covariate-effect parameters from 1,000 observed data sets generated under log-normal frailty settings with $\theta_{LN}=0.8$, hazard ratios ($HR$) of 0.3, 0.5, and 0.8, and sample sizes of $n=500$, 1,000, and 1,500 per group}
	\label{table_4}
	\resizebox{1.0\textwidth}{!}{%
		\begin{threeparttable}
\begin{tabular}{l|c|ccccc}
\hline\hline
$n$ & Baseline parameters & $\beta$ & $HR$ & $\beta$, recovered & $\theta$, recovered \\ 
  \hline 
\multirow{6}{*}{500} & \multirow{3}{*}{$\sigma=0.2$,}  \multirow{3}{*}{$\rho=2.0$}
&  -1.204 & 0.3 & -1.204 (-1.204, -1.139) & 0.960 (0.870, 0.980) \\
& & -0.693 & 0.5 & -0.654 (-0.714, -0.616) & 0.860 (0.820, 0.890) \\
& & -0.223 & 0.8 & -0.198 (-0.248, -0.151) & 0.840 (0.820, 0.880) \\
  \cline{2-6}
 & \multirow{3}{*}{$\sigma=0.4$,}  \multirow{3}{*}{$\rho=1.0$}
 & -1.204 & 0.3 & -1.139 (-1.139, -1.079) & 0.780 (0.740, 0.833)\\
& & -0.693 & 0.5 & -0.693 (-0.734, -0.635) & 0.800 (0.760, 0.860)\\
& & -0.223 & 0.8 & -0.357 (-0.386, -0.223) & 0.820 (0.760, 0.880) \\  \hline 
\multirow{6}{*}{1,000} & \multirow{3}{*}{$\sigma=0.2$,}  \multirow{3}{*}{$\rho=2.0$}
&  -1.204 & 0.3 & -1.204 (-1.204, -1.139) & 0.880 (0.840, 0.900) \\
& & -0.693 & 0.5 & -0.693 (-0.734, -0.635) & 0.880 (0.840, 0.900) \\
& & -0.223 & 0.8 & -0.199 (-0.236, -0.163) & 0.860 (0.800, 0.900) \\
  \cline{2-6}
 & \multirow{3}{*}{$\sigma=0.4$,}  \multirow{3}{*}{$\rho=1.0$}
 & -1.204 & 0.3 & -1.139 (-1.204, -1.139) & 0.780 (0.750, 0.833)\\
& & -0.693 & 0.5 & -0.799 (-0.821, -0.693) & 0.840 (0.790, 0.890)\\
& & -0.223 & 0.8 & -0.223 (-0.248, -0.199) & 0.800 (0.780, 0.840)
\\ \hline
\multirow{6}{*}{1,500} & \multirow{3}{*}{$\sigma=0.2$,}  \multirow{3}{*}{$\rho=2.0$}
&  -1.204 & 0.3 & -1.172 (-1.204, -1.139) & 0.880 (0.820, 0.880) \\
& & -0.693 & 0.5 & -0.693 (-0.734, -0.654) & 0.880 (0.820, 0.900) \\
& & -0.223 & 0.8 & -0.236 (-0.274, -0.198) & 0.880 (0.860, 0.900) \\
  \cline{2-6}
 & \multirow{3}{*}{$\sigma=0.4$,}  \multirow{3}{*}{$\rho=1.0$}
 & -1.204 & 0.3 & -1.172 (-1.204, -1.139) & 0.830 (0.780, 0.850)\\
& & -0.693 & 0.5 & -0.693 (-0.734, -0.654) & 0.840 (0.800, 0.860)\\
& & -0.223 & 0.8 & -0.223 (-0.249, -0.186) & 0.825 (0.800, 0.840)
\\
      \hline \hline
\end{tabular}
		\end{threeparttable}
	}
\end{table}

\subsubsection{Conditional Optimization}

Our simulation results suggest that the ranges of parameter values, especially for $\beta$, that attain the flat maximum surface become wider as the number of parameters being recovered simultaneously increases. This is analogous to likelihood-based estimation, where the objective surface becomes flatter as the number of parameters increases. Accordingly, care is needed when the proposed optimization function is used with an exhaustive grid search over the full parameter set, including the baseline Weibull parameters $\sigma$ and $\rho$.

Therefore, when the flat maximum region is large, we propose a conditional and partially data-dependent approach in which the optimization function ($CC$) is maximized only over $\beta$ and $\theta$, conditional on baseline Weibull parameter estimates obtained from the observed data. For example, these estimates may be obtained from a parametric univariate gamma frailty model \cite{jeong2003}. This strategy introduces the additional assumption that the baseline Weibull parameter estimates from the observed data are close to their true values, at least asymptotically.

\section{Application to NSABP B-14 Data Set}

The proposed approach was applied to the NSABP B-14 data set. Because the true distribution of the frailty term is unknown, we considered both gamma and log-normal working models for the exponentiated omitted covariate. The sample sizes were 1,450 and 1,435 in the placebo and tamoxifen groups, respectively; accordingly, 1,500 subjects per group were used in the TSB simulations. Under the univariate gamma frailty analysis, the estimated inverse frailty variance was $\hat\kappa=0.292$, which corresponded to an estimated omitted covariate effect of $\hat\theta=3.587$ under the assumption that the omitted covariate had unit variance. Under the univariate log-normal frailty analysis, the estimated frailty variance was 0.173, which yielded $\hat\theta=0.374$ from equation (\ref{eqn;6.12}). The estimated treatment effect parameters (log-hazard ratios) from the univariate gamma and log-normal frailty analyses were $-0.738$ and $-0.307$, corresponding to hazard ratios of 0.48 and 0.74, respectively, whereas the naive proportional hazards model that ignored the omitted covariate gave an estimate of $-0.285$ (hazard ratio 0.75).

From Figure \ref{fig:b14_plots}(b), the censoring distribution appears to be similar across the two treatment groups and can be approximated by a two-component mixture of uniform distributions; specifically, $C \sim U(0,a_1)$ with probability $\omega_1$ and $C \sim U(a_1,a_2)$ with probability $1-\omega_1$, where $\omega_1=0.1$, $a_1=14$, and $a_2=21$. This censoring model was adopted in the TSB procedure.

Figure \ref{fig:b14}(a) and (c) show the TSB plots obtained using the parameter estimates from the univariate gamma and log-normal frailty analyses, respectively, overlaid on the observed treatment-specific survival curves from the B-14 data. The corresponding coverage coefficient ($CC$) values were 0.14 and 0.94, suggesting that the parameter estimates from the univariate log-normal frailty model were much closer to the target values than those from the univariate gamma frailty model.

\begin{figure}[htbp]
  \includegraphics[width = 1.0\linewidth,keepaspectratio]{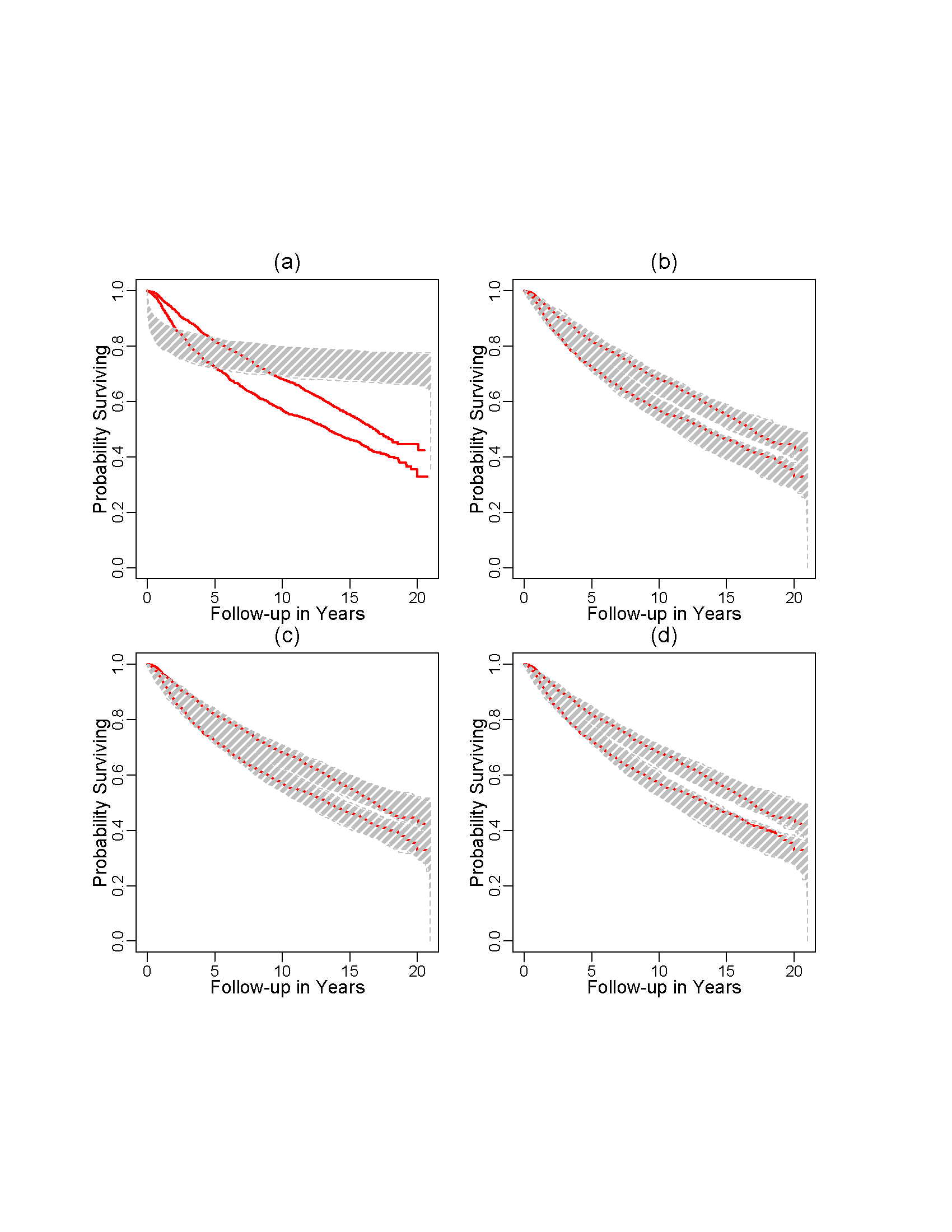}
    \vspace{-1.0in}
    \caption{Illustration of the TSB method using the NSABP B-14 data set: (a) parameter estimates from the univariate gamma frailty analysis, (b) recovered parameter values under the gamma frailty working model, (c) parameter estimates from the univariate log-normal frailty analysis, and (d) recovered parameter values under the log-normal frailty working model.}
    \label{fig:b14}
\end{figure}

Because all model parameters are unknown in the B-14 data, we performed an exhaustive grid search over all parameters simultaneously. Our earlier simulation results suggested that the estimates from the univariate log-normal frailty model provide reasonable starting values. Alternatively, an interactive graphical assessment can help guide the search: for example, a TSB that is too shallow in its early slope suggests $\theta$ should decrease, a band shifted too far to the left suggests $\sigma$ should decrease, and a first event occurring too early suggests $\rho$ should increase. Both approaches indicated that plausible target values were near $\tilde\beta=\log(0.7)$, $\tilde\theta=0.4$, $\tilde\sigma=0.06$, and $\tilde\rho=0.95$. Accordingly, the grid search intervals were set to
\[
\beta \in [\log(0.35),\log(0.9)] \text{ with step size } 0.01,
\]
\[
\theta \in [-0.6,-0.2]\cup[0.2,0.6] \text{ with step size } 0.02,
\]
\[
\sigma \in [0.04,0.08] \text{ with step size } 0.01,
\]
and
\[
\rho \in [0.9,1.0] \text{ with step size } 0.01,
\]
giving a total of 17,967 grid points. We denote this parameter search grid by $\mathcal{P}_n$, to distinguish it
from the time grid $\{g_i\}$ of Section 6.1 over which the survival curves
are evaluated; this notation is used in the consistency result of
Appendix A. The search interval for $\theta$ was chosen to allow for the possibility of a negative omitted covariate effect.

Under the gamma frailty working model, the maximum $CC$ value of 0.97 was attained at a single grid point, yielding recovered parameter values of $\exp(\bar{\hat\beta})=0.68$, that is, $\bar{\hat\beta}=-0.39$, $\bar{\hat\theta}=0.34$, $\bar{\hat\sigma}=0.06$, and $\bar{\hat\rho}=0.98$. Under the log-normal frailty working model, the maximum $CC$ value of 0.97 was attained at five grid points, and the medians of the recovered parameter values over these maximizing grid points were $\exp(\bar{\hat\beta})=0.69$ (plateau interval $=[0.68,0.70]$), that is, $\bar{\hat\beta}=-0.37$, $\bar{\hat\theta}=0.46$ (plateau interval $=[0.42,0.48]$), $\bar{\hat\sigma}=0.06$ (unchanged across all five grid points), and $\bar{\hat\rho}=0.98$ (unchanged across all five grid points). Thus, the recovered treatment effects were very similar under the two frailty assumptions, consistent with the simulation results in Section 6.3. As expected, the main difference between the recovered parameter values under the two frailty assumptions lay in the omitted covariate effect. These values imply frailty variance estimates of 0.109 under the gamma frailty model, obtained by solving $\psi^{\prime}(\kappa)=\bar{\hat\theta}^2$ and taking $1/\kappa$, and 0.291 under the log-normal frailty model, obtained from $\exp(\bar{\hat\theta}^2)\{\exp(\bar{\hat\theta}^2)-1\}$.

Figure \ref{fig:b14}(b) and (d) show the TSB plots obtained using the recovered parameter values under the gamma and log-normal frailty working models, respectively. In both cases, the resulting bands achieved a coverage coefficient of 0.97, confirming that the simulated survival bands closely contained the observed treatment-specific survival curves. For comparison, the analytic formula in equation (\ref{eqn;4.0}) produced a recovered treatment effect of $-1.58$ (hazard ratio 0.21), which is far from the values obtained by the TSB procedure.

\section{Discussion}

As our simulation results showed, accounting for the omitted covariate through univariate frailty analysis can reduce bias reasonably well \emph{on average} when the omitted covariate effect is not large, but its performance deteriorates when that effect is substantial. We therefore proposed a novel simulation-based approach to further reduce bias in the hazard ratio arising from omitted covariates that may induce a vanishing hazard ratio in the observed data. The main idea is to identify the target population parameters that would generate a tightest survival band (TSB) containing the observed treatment-specific survival curves as fully as possible. We provided a formal justification for the empirical behavior of this procedure in Appendix A. Our simulation studies indicate that, under practical censoring scenarios and Weibull baseline event-time distributions, the proposed method performs reasonably well in moderate to large samples. We also illustrated the method using data from a large breast cancer clinical trial.

The proposed approach relies on a few assumptions: (i) unit variance of the omitted covariate, (ii) a random uniform censoring mechanism or a mixture of such mechanisms, (iii) a parametric working model for the frailty term, and (iv) a two-parameter Weibull baseline distribution. These assumptions are practically motivated. First, the omitted covariate may reasonably be regarded as standardized to have unit variance. Second, mixtures of uniform censoring distributions are commonly encountered in randomized clinical trials because of staggered entry and administrative study termination. Third, our simulation results suggest that the TSB approach can recover the target hazard ratio reasonably well even when the assumed frailty distribution is misspecified. Fourth, the two-parameter Weibull distribution provides a flexible yet parsimonious model for the baseline event-time distribution.

Although an exhaustive grid search can be computationally intensive, the burden can be substantially reduced through parallel computing on a high-performance computing system, making the procedure feasible in practice.

\section*{Acknowledgments}
This work utilized the computational resources of the NIH HPC Biowulf cluster (https://hpc.nih.gov).

\appendix
\renewcommand{\thesection}{Appendix A}
\section{Consistency of the TSB Estimator}
\label{app:consistency}


Let $\boldsymbol{\eta} = (\beta, \theta, \sigma, \rho)$ denote the vector of
candidate parameters, and let
$\boldsymbol{\eta}_0 = (\beta_0, \theta_0, \sigma_0, \rho_0)$ denote the true
population values under Model III with a Weibull baseline hazard and
frailty distribution $F_W$ satisfying $\mathrm{Var}(Z) = 1$.

The sample size $n$ is held fixed throughout at the value determined by
the observed data, matching the TSB procedure exactly as used in practice
(Sections 6.1--7): the target parameter $\boldsymbol{\eta}_0$ is
recovered by comparing the observed Kaplan--Meier curve, evaluated at
this fixed $n$, against simulated envelopes generated at this same $n$.
The only limit taken below is in the number of simulation replicates,
$N \to \infty$, which is a design choice available to the analyst (e.g.\
$N=1000$ in Sections 6.2 and 7) and can be made as large as
computationally convenient.

For $x = 0,1$ and each time-grid point $g_i$ $(i=1,\dots,n_g)$, let
$F_{n,x,g_i}^{\boldsymbol{\eta}}$ denote the sampling distribution, at
sample size $n$, of the Kaplan--Meier estimator $\hat S_x(g_i)$ for group
$x$ evaluated at $g_i$, when the data are generated under candidate
parameter $\boldsymbol{\eta}$ (Weibull baseline with parameters
$\sigma,\rho$, treatment effect $\beta$, frailty effect $\theta$, and the
censoring mechanism of Assumption \ref{assump:censoring}). Let
$\ell_n(x,g_i;\boldsymbol{\eta})$ and $u_n(x,g_i;\boldsymbol{\eta})$ denote
the essential infimum and supremum of $F_{n,x,g_i}^{\boldsymbol{\eta}}$,
and write
\[
R_n(x,g_i;\boldsymbol{\eta}) = [\ell_n(x,g_i;\boldsymbol{\eta}),\,
u_n(x,g_i;\boldsymbol{\eta})]
\]
for its essential range. Let $S_x(g_i;\boldsymbol{\eta})$ denote the
population survival value implied by candidate $\boldsymbol{\eta}$ under
the correctly specified Model III survival function (equation
\eqref{eqn;2}), so that $F_{n,x,g_i}^{\boldsymbol{\eta}}$ is centered at,
or asymptotically unbiased for, $S_x(g_i;\boldsymbol{\eta})$; write
$S_x^{(0)}(g_i) \equiv S_x(g_i;\boldsymbol{\eta}_0)$ for the true value.

Let $\mathcal{P}_n$ denote the candidate \emph{parameter} search grid over
$\boldsymbol{\eta}$ described in Section 6.1 and used in the exhaustive
grid search of Section 6.4 -- to be distinguished from the \emph{time}
grid $\{g_i\}_{i=1}^{n_g}$ of Section 6.1, over which survival curves are
evaluated.


\begin{assumption}[Fixed-$n$ range identifiability]
\label{assump:identif}
For every parameter-grid point $\boldsymbol{\eta} \in \mathcal{P}_n$ with
$\boldsymbol{\eta} \neq \boldsymbol{\eta}_0$, there exists at least one
time-grid point $g_i$ and treatment group $x$ such that
\[
R_n(x,g_i;\boldsymbol{\eta}_0) \not\subseteq R_n(x,g_i;\boldsymbol{\eta}),
\]
with $F_{n,x,g_i}^{\boldsymbol{\eta}_0}$ placing positive probability mass
on the part of $R_n(x,g_i;\boldsymbol{\eta}_0)$ lying outside
$R_n(x,g_i;\boldsymbol{\eta})$. The only excluded configurations are
degenerate ones in which treatment has no effect on survival at any
observed time regardless of the omitted-covariate parameter (e.g.\
$\beta = 0$ together with $\theta = 0$), which are not of practical
interest since they correspond to no treatment effect to recover.
\end{assumption}

\noindent
This is a statement about the finite-$n$ \emph{distributions}
$F_{n,x,g_i}^{\boldsymbol{\eta}}$, not merely about their centers
$S_x(g_i;\boldsymbol{\eta})$: it requires that the range of curves
plausibly generated under a wrong $\boldsymbol{\eta}$, at the actual
sample size $n$, misses some non-negligible part of the range of curves
plausibly generated under the truth. This condition can fail -- more
often as $n$ decreases -- when two candidates imply sampling
distributions with heavily overlapping support even though their centers
$S_x(g_i;\boldsymbol{\eta})$ differ; this is the mechanism behind the
flat-maximum-region phenomenon documented empirically in Section 6.3.

The essential range $R_n(x,g_i;\boldsymbol{\eta})$ is the natural object
to use here because it is precisely the quantity that $SPSR_x(g_i)$
targets by construction: $LO_x(g_i)$ and $UP_x(g_i)$ are the sample
min/max of $N$ simulated Kaplan--Meier draws, and sample extrema converge,
as $N \to \infty$, to the essential inf/sup of the generating distribution
rather than to any other feature of it (e.g.\ a quantile or a moment).
Framing Assumption \ref{assump:identif} in terms of $R_n$ therefore
matches the object the TSB procedure actually estimates, rather than
introducing a separate theoretical construct.

\begin{assumption}[Correctly specified censoring]
\label{assump:censoring}
The assumed censoring mechanism (uniform, or a mixture of uniform
distributions, as in Section 6) matches the true censoring distribution,
or is at least non-informative and independent of the treatment
indicator $X$, implying that a grid point arbitrarily close to the truth can always be found
once $n$ is large enough.
\end{assumption}


\begin{proposition}[Fixed-$n$ consistency of the TSB criterion]
\label{prop:fixedn}
Fix the sample size $n$ at the value determined by the observed data.
Under Assumptions \ref{assump:identif}--\ref{assump:censoring}, as the
number of simulation replicates $N \to \infty$,
\[
CC(\boldsymbol{\eta}_0) \;\longrightarrow\; 1 \quad \text{a.s.},
\]
and for every $\boldsymbol{\eta} \in \mathcal{P}_n$ satisfying Assumption
\ref{assump:identif},
\[
CC(\boldsymbol{\eta}) \;\longrightarrow\; c_n(\boldsymbol{\eta}) < 1
\quad \text{a.s.}
\]
Consequently the maximizing set
\[
M_n \;=\; \arg\max_{\boldsymbol{\eta} \in \mathcal{P}_n} CC(\boldsymbol{\eta})
\;=\;
\left\{ \boldsymbol{\eta} \in \mathcal{P}_n :
CC(\boldsymbol{\eta}) = \max_{\boldsymbol{\eta}' \in \mathcal{P}_n} CC(\boldsymbol{\eta}') \right\}
\]
contains $\boldsymbol{\eta}_0$ with probability tending to 1 as $N \to
\infty$, at the fixed sample size $n$. $M_n$ need not be a singleton at
finite $n$: any additional grid points it contains must, by Assumption
\ref{assump:identif}, be indistinguishable in range from
$\boldsymbol{\eta}_0$ at that sample size and grid resolution -- precisely
the flat-maximum-region phenomenon of Section 6.3.
\end{proposition}

\begin{proof}[Proof]

\textbf{Step 1: The simulated envelope resolves the true finite-$n$
sampling range.}
Fix $\boldsymbol{\eta} \in \mathcal{P}_n$, a time-grid point $g_i$, and a
group $x$. The simulated Kaplan--Meier values
$\hat S^{(sim)}_{x,j}(g_i)$, $j=1,\dots,N$, are i.i.d.\ draws from
$F_{n,x,g_i}^{\boldsymbol{\eta}}$. By almost-sure convergence of sample
extrema to the essential inf/sup of the generating distribution,\footnote{This
is a standard fact: writing $u = u_n(x,g_i;\boldsymbol{\eta})$ for the
essential supremum, for any $\epsilon>0$ we have
$\Pr[UP_x(g_i) < u-\epsilon] = F_{n,x,g_i}^{\boldsymbol{\eta}}(u-\epsilon)^N
\to 0$ as $N \to \infty$, since
$F_{n,x,g_i}^{\boldsymbol{\eta}}(u-\epsilon) < 1$ by definition of the
essential supremum; since $UP_x(g_i)$ is non-decreasing in $N$ and bounded
above by $u$, this forces $UP_x(g_i) \to u$ a.s. The argument for
$LO_x(g_i) \to \ell_n(x,g_i;\boldsymbol{\eta})$ is symmetric. See, e.g.,
Galambos (1987, \emph{The Asymptotic Theory of Extreme Order Statistics},
2nd ed., Ch.\ 2) or Leadbetter et al.\ (1983, \emph{Extremes and Related
Properties of Random Sequences and Processes}, Springer, Sec.\ 1.2) for
this and related results on convergence of sample extrema.}
\[
LO_x(g_i) \to \ell_n(x,g_i;\boldsymbol{\eta}), \qquad
UP_x(g_i) \to u_n(x,g_i;\boldsymbol{\eta}) \qquad \text{a.s.\ as } N \to \infty,
\]
so $SPSR_x(g_i) \to R_n(x,g_i;\boldsymbol{\eta})$ a.s. This occurs at the
fixed sample size $n$; increasing $N$ only reveals the true support of a
fixed, finite-$n$ sampling distribution, and does not require $n$ to grow.

\textbf{Step 2: Containment at the truth.}
Set $\boldsymbol{\eta} = \boldsymbol{\eta}_0$. The observed curve
$\hat S_x^{(obs)}(g_i)$ is itself a single draw from
$F_{n,x,g_i}^{\boldsymbol{\eta}_0}$ -- the same distribution generating the
simulated replicates in Step 1 at $\boldsymbol{\eta}_0$ -- so
$\hat S_x^{(obs)}(g_i) \in R_n(x,g_i;\boldsymbol{\eta}_0)$ with probability
1. Combining with Step 1, the containment indicator at $(g_i,x)$
converges to 1 a.s.\ as $N \to \infty$, for every grid point and group.
Hence $f_x \to n_g$ for each $x$, and
$CC(\boldsymbol{\eta}_0) \to 1$ a.s. This uses only the fixed sample size
$n$ of the observed data; no limit in $n$ is required.

\textbf{Step 3: Reduced containment away from the truth.}
Consider $\boldsymbol{\eta} \neq \boldsymbol{\eta}_0$ satisfying
Assumption \ref{assump:identif}, so at some $(g_i,x)$ the true range
$R_n(x,g_i;\boldsymbol{\eta}_0)$ is not contained in
$R_n(x,g_i;\boldsymbol{\eta})$, with $F_{n,x,g_i}^{\boldsymbol{\eta}_0}$
placing positive mass $p_{g_i,x}>0$ on the excluded region. By Step 1,
$SPSR_x(g_i) \to R_n(x,g_i;\boldsymbol{\eta})$ a.s. Since
$\hat S_x^{(obs)}(g_i) \sim F_{n,x,g_i}^{\boldsymbol{\eta}_0}$,
\[
\Pr\!\left[\hat S_x^{(obs)}(g_i) \notin SPSR_x(g_i)\right]
\longrightarrow p_{g_i,x} > 0 \quad \text{as } N \to \infty,
\]
so the containment indicator at $(g_i,x)$ does not converge to 1.
Aggregating over the finite grid and both groups, and noting
$\mathcal{P}_n$ is finite at fixed $n$, gives
$CC(\boldsymbol{\eta}) \to c_n(\boldsymbol{\eta}) < 1$ a.s., for every
identifiable $\boldsymbol{\eta} \in \mathcal{P}_n \setminus \{\boldsymbol{\eta}_0\}$.

\textbf{Step 4: Conclusion.}
Steps 2--3 show that, at the fixed sample size $n$, $CC(\boldsymbol{\eta}_0)
\to 1$ while $CC(\boldsymbol{\eta}) \to c_n(\boldsymbol{\eta}) < 1$ for
every identifiable competitor, as $N \to \infty$. Hence, for $N$
sufficiently large, $\boldsymbol{\eta}_0$ attains a strictly higher
coverage coefficient than any identifiable competitor on the grid, and
$\Pr[\boldsymbol{\eta}_0 \in M_n] \to 1$ as $N \to \infty$. Grid points
violating Assumption \ref{assump:identif} may remain in $M_n$ alongside
$\boldsymbol{\eta}_0$; this is the flat-maximum-region phenomenon of
Section 6.3, and is a genuine feature of the finite-$n$ identification
problem rather than an artifact of the proof.
\end{proof}

\bibliographystyle{elsarticle-num}

\end{document}